\documentclass{article}
\usepackage{spconf,amsmath,graphicx,bm,amssymb}
\usepackage{algorithm}
\usepackage{color}
\usepackage{amsthm}

\usepackage{graphicx,subfigure}

\usepackage[pdfstartview=FitH,
CJKbookmarks=true,
bookmarksnumbered=true,
bookmarksopen=true,
colorlinks,
pdfborder=001,
linkcolor=blue,
anchorcolor=blue,
citecolor=blue,
]{hyperref}
\hypersetup{hidelinks}

\usepackage{algorithmic}
\newcommand{\bb}{\mathbf{b}}

\newcommand{\bh}{\mathbf{h}}

\newcommand{\bn}{\mathbf{n}}

\newcommand{\bx}{\mathbf{x}}

\newcommand{\br}{\mathbf{r}}
\newcommand{\bs}{\mathbf{s}}

\newcommand{\by}{\mathbf{y}}
\newcommand{\bz}{\mathbf{z}}

\newcommand{\bH}{\mathbf{H}}

\newtheorem{remark}{Remark~}%[section]
\newtheorem{assumption}{Assumption~}%[section]
\newtheorem{definition}{Definition~}%[section]
\newtheorem{theorem}{Theorem~}%[section]
\newtheorem{lemma}{Lemma~}%[section]
\newtheorem{proposition}{Proposition~}%[section]
\newtheorem{corollary}{Corollary~}%[section]

\def\x{{\mathbf x}}

\title{An AMP-Based Asymptotic Analysis for nonlinear one-bit precoding}%
\name{Zheyu Wu$^{\S}$,  Junjie Ma$^{\dag}$, Ya-Feng Liu$^{\star}$,  and Bruno Clerckx$^{\S}$
}

\address{$^{\S}$Department of Electrical and Electronic Engineering, Imperial College London, London, U.K. \\[2pt]
    $^{\dag}$LSEC, ICMSEC, AMSS, Chinese Academy of Sciences, Beijing, China\\[2pt]
$^{\star}$School of Mathematical Sciences, Beijing University of Posts and Telecommunications, Beijing, China\\[2pt]
 Email:  \{zheyu.wu, b.clerckx\}@imperial.ac.uk, majunjie@lsec.cc.ac.cn, yafengliu@bupt.edu.cn}
  
\begin{document}

\ninept
\maketitle

\begin{abstract}
This paper focuses on the asymptotic analysis of a class of nonlinear one-bit precoding schemes  under Rayleigh fading channels. The considered scheme employs a convex-relaxation-then-quantization (CRQ) approach to the well-known minimum mean square error (MMSE) model, which includes the classical one-bit precoder SQUID as a special case. 
 %The considered scheme is based on the well-known minimum mean square error (MMSE) model and 
 %encompasses the classical one-bit precoder SQUID \cite{SQUID} as a special case.  %We propose a novel analytical framework based on approximate message passing, and derive closed-form expressions for the symbol error probability (SEP) of the considered scheme in the large-system limit. Our result quantitatively characterize how the system parameters and regularization parameters in the considered model affect the SEP performance. Simulation results verify our analytical result and also demonstrate that performance gain over SQUID can be achieved by appropriately tuning the regularization parameters in the considered model.  
 To analyze its asymptotic behavior, we develop a novel analytical framework based on approximate message passing (AMP). We show that, the statistical properties of the considered scheme can be asymptotically characterized by a scalar ``signal plus Gaussian noise'' model. Based on this, we further derive a closed-form expression for the symbol error probability (SEP) in the large-system limit, which quantitatively characterizes the impact of both system and model parameters on SEP performance. Simulation results validate our analysis and also demonstrate that performance gains over SQUID can be achieved by appropriately tuning the parameters involved in the considered model.
 %, in which a large number of devices sporadically transmit small data to a base-station (BS) by sending preassigned non-orthogonal signature sequences % widely used in the existing literature and 

%A convergence property is established.

%achieves the same detection performance as the state-of-the-art algorithms while having
%, whose computational complexity is much less than solving the detection problem for all sequences.

%achieves similar performance as some existing algorithms while having better computational scalability.
%
%The problem is formulated as a quadratic programming problem. Various semidefinite relaxation (SDR) based sub-optimal
%algorithms have been proposed to solve the problem. In this paper, we propose an enhanced SDR to general symbol constellation cases,
%and give a comprehensive comparison between the proposed SDR method and previous ones. Some theoretical results are also provided
%to show when these different SDRs are tight. Numerical simulations show the enhanced
%SDR can provide a comparable lower bound with some stronger SDRs, but can be solved much more efficiently.
%
%[[[Will rewrite the abstract!!!]]]
\end{abstract}
\begin{keywords}
Approximate message passing, asymptotic analysis, nonlinear one-bit precoding.\end{keywords}
\vspace{-0.1cm}
\section{Introduction}
\vspace{-0.1cm}
To meet the growing demand for high capacity, reliability, and connectivity in wireless communication systems, antenna arrays have continued to scale up in size, giving rise to massive multi-input multi-output (MIMO) systems \cite{massivemimo1}.  However, conventional radio frequency (RF) chains equipped with high-resolution digital-to-analog converters (DACs) incur prohibitive hardware cost and power consumption, which severely limits the scalability of antenna arrays and poses a significant implementation challenge. To address this issue, one-bit DACs have emerged as an attractive alternative for realizing energy-efficient and cost-effective massive MIMO systems.  In recent years, one-bit precoding has attracted significant research interests \cite{SQUID,MFrate,ZF,MMSE2,WSR,duality_onebit,SEPlinear,C3PO2,ADMM,CIfirst,sep2,CImodel,GEMM,PBB,onebit_CI,qce_isac}. 

%The most straightforward one-bit precoding approach is linear-quantized precoding, where the one-bit transmit signal is obtained by directly quantizing the output of a conventional linear precoder. Extensive research has been devoted to both the design and analysis of linear-quantized precoding \cite{}. However, the coarse one-bit quantization inevitably results in significant performance degradation, as shown in the existing literature. These limitations have motivated the investigation of nonlinear symbol-level precoding.% By restricting each DAC output to only two levels, they substantially reduce power consumption and hardware cost. 
The most straightforward one-bit precoding approach is linear-quantized precoding, where the one-bit transmit signal is obtained by directly quantizing the output of a conventional linear precoder. To capture the effect of one-bit quantization,  new analytical frameworks based on Bussgang decomposition \cite{Bussgang} and random matrix theory \cite{SEPlinear} have been developed. Building on these frameworks, extensive research has been devoted to both the design and analysis of linear-quantized precoding \cite{MFrate,ZF,MMSE2,WSR,duality_onebit,SEPlinear}. However, the coarse one-bit quantization inevitably results in significant performance degradation, as demonstrated in existing literature. This has motivated the investigation of nonlinear one-bit precoding.

Nonlinear (symbol-level) one-bit precoding generates the one-bit transmit signal based on both the channel state information and the instantaneous data symbols at each time slot, typically by solving a dedicated optimization problem with explicit one-bit constraints. It generally achieves significantly better performance than linear-quantized precoding, albeit at the cost of higher computational complexity. Recent research on nonlinear (one-bit) precoding has mainly focused on developing efficient and effective optimization algorithms \cite{SQUID,C3PO2,ADMM,CIfirst,sep2,CImodel,GEMM,PBB,onebit_CI,qce_isac}, while its performance analysis remains largely underexplored. 

{\color{black}In this paper, we provide a rigorous asymptotic analysis of a class of nonlinear one-bit precoding schemes. The considered scheme is based on a convex-relaxation-then-quantization (CRQ) approach applied to the minimum mean square error (MMSE) model widely adopted in the literature, which includes the classical one-bit precoder SQUID \cite{SQUID} as a special case. We develop a novel analytical framework based on approximate message passing (AMP) \cite{donoho2009message,AMP,feng2022unifying,LASSO,elasticnet,ma2017orthogonal,rangan2019vector,fan2022approximate,dudeja2023universality,liu2024unifying}. Building on this framework, we provide a precise characterization of the asymptotic behavior of the considered scheme.} %which enables a precise characterization for the asymptotic behavior of the considered scheme.%which enables us to 
%characterize the asymptotic behavior of the considered scheme.% and quantifies the impact of both system and model parameters on system performance. 

\vspace{-0.2cm}
\section{System model and problem formulation}\label{sec:2}
\vspace{-0.2cm}

We consider a downlink multiuser multi-input single-output (MU-MISO) system consisting of an $N$-antenna base station (BS) and $K$ single-antenna users. The BS is equipped with one-bit DACs and simultaneously transmits data symbols to all users, while the users are assumed to have infinite-resolution anolog-to-digital converters (ADCs). The received signal vector at the users can then be expressed as\footnote{Following related literature \cite{ma2025,PAPR}, we consider a real-valued system in this paper.} $$\by=\bH\x_T+\bn,$$ 
where $\bH\in\mathbb{R}^{K\times N}$ is the channel matrix between the BS and the users, $\x_T\in\{-1,1\}^N$ is the transmit signal vector at the BS constrained by one-bit DACs, and   $\bn\sim\mathcal{N}(0,\sigma^2 \mathbf{I})$ is the additive white Gaussian noise (AWGN). 

This paper focuses on the nonlinear one-bit precoding strategy, where the transmit signal $\x_T$ is designed based on both the channel state information $\bH$ and the instantaneous data symbol vector, denoted by $\bs\in\mathbb{R}^{K}$. Following the practice in one-bit precoding literature, we assume that the users rescale the received signal by a factor $\xi\in\mathbb{R}$ to remove the effect of channel gain, yielding the symbol estimate $\tilde{\bs}=\xi\by$. Taking the mean square error (MSE) between estimated and true data symbols, i.e., $$\frac{1}{N}\mathbb{E}_{\bn}[\|\tilde{\bs}-\bs\|^2]=\frac{1}{N}\|\bs-\xi\bH\x_T\|^2+\frac{\sigma^2K}{N}\xi^2,$$ as the performance metric, the MMSE-optimal one-bit precoder can then be obtained by solving 
\begin{equation}\label{eq:MSE}
\min_{\xi,\x_T\in\{-1,1\}^N}\frac{1}{N}\|\bs-\xi\bH\x_T\|^2+\frac{\sigma^2 K}{N}\xi^2.
\end{equation}
The above problem is a large-scale mixed-integer programming and is extremely challenging to solve. In this paper, we focus on a convex-relaxation-then-quantization (CRQ) scheme \cite{SQUID,C3PO2,ADMM}, which first solves an appropriate convex relaxation of \eqref{eq:MSE} and then quantizes the solution to satisfy the one-bit constraint. 
 %A computationally tractable one-bit solution is typically obtained by a convex-relaxation-then-quantization (CRQ) framework, which first solves an appropriate convex relaxation of \eqref{eq:MSE} and then quantizes the solution to satisfy the one-bit constraint.  In the following, we introduce the specific CRQ precoding scheme considered in this paper. 

Specifically, introducing $\x=\xi\x_T$ and $a=|\xi|$,  problem \eqref{eq:MSE} becomes
\begin{equation}\label{eq:MSE2}
\begin{aligned}
\min_{a,\x}~&\frac{1}{N}\|\bs-\bH\x\|^2+\frac{\sigma^2 K}{N}a^2\\
\text{s.t. }~& |x_i|= a,~i=1,2,\dots, N.
\end{aligned}
\end{equation}
The introduction of $\x=\xi\x_T$ follows \cite{SQUID}, which transforms the objective function into a convex form by eliminating the multiplicative variable. The reformulation in \eqref{eq:MSE2} motivates the following CRQ precoding scheme. 
\begin{definition}[CRQ precoding]\label{def:CRQ}
The transmit signal is $\x_T=\text{sgn}(\hat{\x})$, where $\hat{\x}$ (together with $\hat{a}=\|\hat{\x}\|_\infty$) is the solution to the following convex relaxation of \eqref{eq:MSE2}:
\begin{equation}\label{eq:MSErelax}
\begin{aligned}
\min_{a,\x}~&\frac{1}{N}\|\bs-\bH\x\|^2+\frac{\rho}{N}\|\x\|^2+\lambda a^2\\
\text{s.t. }~& |x_i|\leq a,~i=1,2,\dots, N,
\end{aligned}
\end{equation}
where $\rho>0$ and $\lambda>0$ are regularization parameters.
\end{definition}
\begin{remark}
The model in \eqref{eq:MSErelax} is equivalent to 
$$\min_{\x}~\frac{1}{N}\|\bs-\bH\x\|^2+\frac{\rho}{N}\|\x\|^2+\lambda \|\x\|_\infty^2.$$
Hence, the well-known one-bit precoder SQUID \cite{SQUID} is a special case of the considered CRQ scheme with $\rho=0$ and $\lambda=\frac{\sigma^2K}{N}$.
\end{remark}
The model in \eqref{eq:MSErelax} relaxes the non-convex constraint  in \eqref{eq:MSE2} to a convex boxed constraint and introduces a general regularization parameter for the term $a^2$. In addition, an $\ell_2$ regularization is added to make the objective function strongly convex in $\x$, which is important for the analysis and also provides extra flexibility for the model.
Our goal in this paper is to theoretically characterize the performance of the CRQ precoding scheme given by Definition \ref{def:CRQ}.%, and then leverage the analytical result to optimize $(\lambda,\rho)$. 
\vspace{-0.15cm}
 \section{Main Reults}
\subsection{Preliminaries}\label{Sec:assumptions}
Our asymptotic analysis is based on the following assumptions.
\vspace{-0.1cm}
\begin{assumption}\label{ass}
The following assumptions hold throughout the paper:

(i) %The number of transmit antennas $N$ and the number of users $K$ tend to infinity at the same rate, i.e., $
$N, K\to \infty \text{ with } \lim_{N,K\to\infty}\frac{K}{N}=\delta\in(0,\infty)$;

(ii) The entries of  $\bH$ are independently drawn from $\mathcal{N}(0,\frac{1}{K})$;

(iii) The entries of $\bs$ are independently drawn from binary shift keying (BPSK) constellation, i.e., $s_k\sim\text{Unif}(\{-1,1\})$.
 \end{assumption}
 
 The first two assumptions are standard in asymptotic analysis. The last assumption is imposed for simplicity of discussions; our approach can be easily generalized to high-order modulation. 
 
 We next give two lemmas that are important for presenting our main results; their proofs are omitted due to the space limitation. 
\begin{lemma}\label{lemma:unique}
Given $\rho>0$,  $\delta>0$,  $a\geq 0$, and define %$\eta_a:\mathbb{R}\times\mathbb{R}_{>0}\to\mathbb{R}$ as
%\begin{equation}\label{def:etaa}
$\eta_a(x;\gamma)=\mathcal{P}_{[-a,\,a]}(\frac{x}{\gamma+1}),$
%\end{equation}
where $\mathcal{P}_{[-a,a]}(\cdot)$ denotes the projection onto $[-a,a]$.
There exists a unique solution $(\tau_a^2,\gamma_a)$ to the following equations:
\begin{subequations}\label{condition0}
\begin{align}
\tau^2&=1+{\delta}^{-1}\mathbb{E}\left[\eta_a^2(\tau Z ;\gamma)\right],\label{conditiona}\\
\rho&=\gamma(1-{\delta}^{-1}\eta_a'(\tau Z;\gamma))%\mathbb{P}\left(\frac{\tau Z}{\gamma+1}\in[-a,a]\right)\right),
\label{condition0b}
\end{align}
\end{subequations} 
where $Z\sim\mathcal{N}(0,1)$ and $\eta_a'(x;\gamma)$ takes the partial derivative with respect to its first entry.  
\end{lemma}
\begin{lemma}\label{lemma:astar}
Given $\rho>0$, $\lambda>0$, and $\delta>0$, 
define %the function $f:\mathbb{R}_{\geq0}\to\mathbb{R}$ as
%\begin{equation}\label{phia}
${f}(a)=\delta\rho(\tau_a^2-1)+{\delta\rho^2\tau_a^2}{\gamma_a^{-2}}+\lambda a^2,$
%\end{equation} 
where $(\tau_a^2,\gamma_a)$ is the solution to \eqref{condition0}.
Then $f(a)$ is strongly convex and continuously differentiable on $[0,\infty)$. Furthermore, the unique minimizer of $f(a)$ over $[0,\infty)$, denoted by $a^*$, is positive. 
\end{lemma}
\vspace{-0.5cm}

\subsection{Main Results}\label{subsec:mainresult}
The following theorem characterizes  the limiting empirical distribution of $(\bH\bx_T, \bs)$, where $\x_T$ is defined in Definition \ref{def:CRQ}.
\begin{theorem}\label{theorem1}
Let  $(\tau_*^2,\gamma_*):=(\tau_{a^*}^2,\gamma_{a^*})$ be the solution to \eqref{condition0} with $a=a^*$, where $a^*$ is defined in Lemma \ref{lemma:astar}. 
Consider the CRQ precoding scheme given by Definition \ref{def:CRQ}. 
 Under Assumption \ref{ass}, the following result holds for any pseudo-Lipschitz  function $\psi:\mathbb{R}^2\rightarrow \mathbb{R}$: \vspace{-0.15cm}
  \begin{equation}\label{mainresult:sgnfunction}
  \lim_{K\to\infty}\frac{1}{K}\sum_{k=1}^K\psi(\bh_k^T {\x}_T,s_k)\overset{a.s.}{=}\mathbb{E}\left[\psi(\bar{\alpha}S+\bar{\beta}^{\frac{1}{2}}Z,S)\right],
  \end{equation}
  where  %\begin{equation}\label{Eqn:alpha_beta}
%\begin{aligned}
%\end{aligned}
%\end{equation}
  $S\sim\text{\normalfont{Unif}}(\{-1,1\})$, ${Z}\sim\mathcal{N}(0,1)$ is independent of $S$, $\bar{\alpha}=\sqrt{\frac{2}{\pi}}\frac{1}{\delta\tau_*},~\bar{\beta}={\delta}^{-1}\mathbb{E}[(\bar{\alpha}|\hat{X}|-1)^2],$  and $\hat{X}=\eta_{a^*}\left({\tau_*Z};{\gamma_*}\right).$
  \end{theorem}
  \begin{proof}
  See Section \ref{sec:proof}.
  \end{proof}
  Theorem \ref{theorem1} implies that the statistical properties of the considered system model in Section \ref{sec:2}, 
with the CRQ precoding scheme in Definition \ref{def:CRQ}, can be asymptotically characterized  by a scalar ``signal plus Gaussian noise'' model as \vspace{-0.2cm}
\begin{equation}\label{asymptotic}
y=\bar{\alpha} S+\sqrt{\bar{\beta}+\sigma^2}Z,
\end{equation} where $\sigma^2$ is the variance of the AWGN. %$S\sim\text{Unif}(\{-1,1\})$ and $Z\sim\mathcal{N}(0,1)$ are independent,  $(\bar{\alpha},\bar{\beta})$ are deterministic constants given by Theorem \ref{theorem1}, and $\sigma^2$ denotes the channel noise variance.
As a direct corollary, we can characterize the symbol error probability (SEP) achieved by CRQ precoding. The SEP is a widely used performance metric in the one-bit precoding literature, defined as 
$\text{SEP}=\frac{1}{K}\sum_{k=1}^K\mathbb{P}\left(\hat{s}_k\neq s_k\right),$
where $s_k$ and $\hat{s}_k$ are the intended data symbol and the detected symbol at user $k$, respectively. In particular, we consider the symbol detector $\hat{s}_k=\text{sgn}\left(\bh_k^T{\x}_T+n_k\right),$~$k=1,2,\dots, K$, for BPSK constellation.   The following result shows that the SEP of the CRQ precoding converges to that of the asymptotic model in \eqref{asymptotic}.
\begin{corollary}\label{SEP}
The following result holds under Assumption \ref{ass}:\vspace{-0.1cm}
$$\lim_{K\to\infty}\frac{1}{K}\sum_{k=1}^K\mathbb{P}\left(\hat{s}_k\neq s_k\right)\overset{a.s.}{=}Q\left(\sqrt{\overline{\text{\normalfont{SNR}}}}\right),$$
where $Q(x)$ is the  tail distribution function of the standard Gaussian distribution, and $\overline{\text{\normalfont SNR}}=\bar{\alpha}^2/(\bar{\beta}+\sigma^2)$ is the signal to noise ratio (SNR) of the asymptotic model in \eqref{asymptotic}.
\end{corollary}

Corollary \ref{SEP} provides insights for selecting the regularization parameters $(\rho,\lambda)$ in model \eqref{eq:MSErelax}. In principle, they can be optimized to minimize the asymptotic SEP. However, the dependence of the asymptotic SEP on $(\rho,\lambda)$ is implicit and  complicated, resulting in a challenging two-dimensional optimization problem. Analytically deriving the optimal $(\rho,\lambda)$ is left as an interesting future work.

{\color{black} Before concluding this section, we briefly discuss a very recent related work \cite{ma2025}, which provides an asymptotic analysis of a quantized box-constrained one-bit precoding scheme. Our work differs from \cite{ma2025} in both the considered scheme and the analytical framework.
 In \cite{ma2025}, the variable $a$ in \eqref{eq:MSErelax} is treated as a fixed parameter that must be carefully tuned, whereas our scheme jointly optimizes $a$ and $\x$, thus improving the performance and eliminating the need of  manual tuning. %while ensuring satisfactory performance.
 In terms of the analytical framework,  \cite{ma2025} is based on a convex Gaussian min–max theorem  \cite{CGMT} specifically tailored to the sign function. In contrast,  our analysis is built upon AMP, as elaborated in the next section. Our proposed framework is more versatile and can be readily generalized to capture broader nonlinear effects beyond one-bit quantization.}
\vspace{-0.2cm}
 
\section{Proof Sketch of Theorem 1}\label{sec:proof}
\vspace{-0.1cm}
  In this section, we outline the key steps involved in the proof of Theorem \ref{theorem1}. 
  The main technique used in our proof is AMP \cite{donoho2009message,AMP}, which defines an iterative process as follows:\vspace{-0.1cm}
\begin{equation}\label{AMP_form}
\begin{aligned}
\bx_{t+1}&=\eta_t(\bx_{t}+\bH^T\bz_t),\\
\bz_t&=\bs-\bH\bx_t+\frac{1}{\delta}\left<\eta_{t-1}'(\bx_{t-1}+\bH^T\bz_{t-1})\right>\bz_{t-1},
\end{aligned}\vspace{-0.1cm}
\end{equation}
where $\eta_t(\cdot)$ is a sequence of Lipschitz denoising functions that act element-wise on their inputs, $\left<\cdot\right>$ takes the element-wise average of its input, and $\bH$ is a random matrix satisfying Assumption \ref{ass} (ii). A key feature of AMP is that its asymptotic behavior can be rigorously characterized via a one-dimensional recursion known as the state evolution. In particular, as an important application of AMP, the denoising function can be specifically designed such that the corresponding AMP algorithm converges to the solution of a convex optimization problem. It is then widely used to characterize the high-dimensional statistics of convex optimization solutions \cite{LASSO,elasticnet}.  

Compared to existing AMP-based analyzes, there are two unique challenges of our problem. First, the variable $a$ is coupled with all elements in $\x$ through the box constraint in \eqref{eq:MSErelax}, making it difficult to design a corresponding AMP algorithm with separable denoising functions. Second, characterizing $\bH \text{sgn}(\hat{\x})$ requires the correlation information between $\bH$ and $\hat{\x}$, which is more complicated than analyzing only the distribution of $\hat{\x}$. To address the above two challenges, our analytical framework contains two steps. 

\emph{Step 1: Formulate an auxiliary problem with separable constraints.}
 First, we show that the empirical statistics of $\bH\text{sgn}(\hat{\x})$ are asymptotically equivalent to those of  $\bH\text{sgn}({\x}^*)$, where $\x^*$ is the solution to an auxiliary problem with separable constraints.%; see the following proposition.
 \begin{proposition}\label{the1}
Under Assumption \ref{ass}, the following result holds for any pseudo-Lipschitz function $\psi:\mathbb{R}^2\rightarrow \mathbb{R}$:\vspace{-0.15cm}
$$\lim_{K\to\infty}\left|\frac{1}{K}\sum_{k=1}^K\psi(\bh_k^T\text{sgn}(\hat{\x}), s_k)-\frac{1}{K}\sum_{k=1}^K\psi(\bh_k^T\text{sgn}({\x}^*),s_k)\right|\overset{a.s.}{=}0,$$
where \vspace{-0.25cm}
%\begin{equation}
$$\x^*=\arg\min_{\x\in[-a^*,a^*]^N}\left\{\frac{1}{N}\|\bs-\bH\x\|^2+\frac{\rho}{N}\|\x\|^2\right\},\vspace{-0.1cm}$$
%\end{equation}
 $\hat{\x}$ is the solution to \eqref{eq:MSErelax},  and  $a^*$ is defined in Lemma \ref{lemma:astar}.
\end{proposition}

\emph{Step 2: Analyze the auxiliary problem.} We next focus on the auxiliary problem defined in Proposition \ref{the1}. Specifically,  we prove the following result for $\x^*$.
\vspace{-0.1cm}
 \begin{proposition}\label{pro:convergexa}
 Under Assumption \ref{ass}, the following result holds:\vspace{-0.15cm}
$$\lim_{K\to\infty}\frac{1}{K}\sum_{k=1}^K \psi(\bh_k^T\text{sgn}(\bx^*), s_k)=\mathbb{E}[\psi(\bar{\alpha}S+\bar{\beta}^{\frac{1}{2}}Z,S)],\vspace{-0.1cm}$$
where  $\bar{\alpha},\bar{\beta}, S,$ and $Z$ are defined in Theorem \ref{theorem1}.
\end{proposition}
 
Combining Propositions \ref{the1} and \ref{pro:convergexa} gives Theorem \ref{theorem1}.   In the following, we introduce the main ideas for proving Propositions \ref{the1} and \ref{pro:convergexa} using AMP.
  \vspace{-0.15cm}
  \subsection{Proof of Proposition \ref{the1}}\vspace{-0.05cm}
  Note that when $a$ is fixed, the constraint in \eqref{eq:MSErelax} becomes separable. This motivates a nested formulation of \eqref{eq:MSErelax}, with an outer optimization over  $a$ and an inner problem over $\x$ with a fixed $a$, as follows: 
\begin{equation}\label{def:aN}
\min_{a\geq 0}f_N(a;\bs,\bH)+\lambda a^2,
\end{equation} where \vspace{-0.1cm}
\begin{equation}\label{def:fN}
f_N(a;\bs,\bH)=\min_{\x\in[-a,a]^N}~\left\{\frac{1}{N}\|\bs-\bH\x\|^2+\frac{\rho}{N}\|\x\|^2\right\}.
\end{equation}
Let $\x_a$ be the solution to \eqref{def:fN} with a fixed $a\geq 0$. Under the above notations, $\x^*$ defined in Proposition \ref{the1} can be written as $\x^*=\x_{a^*}$, and the optimal solution to \eqref{eq:MSErelax} can be expressed as $(\hat{a}_N,\x_{\hat{a}_N})$, where $\hat{a}_N$ is the solution to \eqref{def:aN}.

The optimization problem in \eqref{def:fN} has a separable constraint and is therefore amenable to analysis using AMP.  Specifically, we can show that the AMP algorithm corresponding to \eqref{def:fN} is obtained by setting $\eta_t(x)=\eta_a(x;\gamma_a)$ in \eqref{AMP_form}, where $\eta_a$ and $\gamma_a$ are defined in Lemma \ref{lemma:unique}, and the state evolution condition of the AMP algorithm is given by \vspace{-0.1cm}
\begin{equation}\label{SE}
\tau_{t+1}^2=1+\mathbb{E}[\eta_a^2(\tau_tZ;\gamma_a)],~~Z\sim\mathcal{N}(0,1),
\end{equation} 
where $\tau_0^2=1$. 
 In particular, the condition \eqref{conditiona} is the fixed point condition of \eqref{SE}, and the condition \eqref{condition0b} relates the limiting point of the AMP algorithm to the first-order optimality condition of problem \eqref{def:fN}. %The derivation follows similar ideas as in \cite{LASSO,elasticnet}, and details are omitted due to space limitations. 
Following a procedure similar to that in \cite{LASSO,elasticnet}, we get the following result, which characterizes the convergence of the  AMP algorithm, as well as the asymptotic properties of the optimal solution and value of \eqref{def:fN}. 
\vspace{-0.05cm}
\begin{lemma}\label{prob:converge}
Let $(\tau_a^2,\gamma_a)$ be the solution to \eqref{condition0} with a fixed $a\geq 0$, and let $\{(\x_{t+1},\bz_t)\}_{t\geq 0}$ be the sequence generated by the AMP algorithm with $\eta_t(x)=\eta_a(x;\gamma_a)$ and $\x_0=\mathbf{0}$. 
The following results hold under Assumption \ref{ass}:
\begin{enumerate}
\item[(i)] $ \lim_{t\to\infty}\lim_{N\to\infty}\frac{1}{N}\|\x_t-\x_a\|^2\overset{a.s.}{=}0$;
\item[(ii)] $\lim_{N\to\infty}\frac{1}{N}\sum_{i=1}^N\varphi(x_{a,i})\overset{a.s.}{=}\mathbb{E}\left[\varphi(X_a)\right]$ for any pseudo-Lipschitz function $\varphi$, 
where %$x_{a,i}$ denotes the $i$-th element of $\x_a$, 
%\begin{equation}\label{def:Xa}
$X_a=\eta_a(\tau_aZ;\gamma_a)$ with $Z\sim\mathcal{N}(0,1)$;
%\end{equation}
\item[(iii)]$\lim_{N\to\infty} f_{N}(a;\bs,\bH)\overset{a.s.}{=}\delta\rho(\tau_a^2-1)+\delta\rho^2\tau_a^2\gamma_a^{-2}.$
%where $f_N(a;\s,\H)$ is defined in \eqref{def:fN},  
%$$\bar{f}(a):=\delta\rho(\tau_a^2-1)+\frac{\delta\rho^2\tau_a^2}{\gamma_a^2},$$and
\end{enumerate}
\end{lemma}
As implied by Lemma \ref{prob:converge}, the objective function in problem \eqref{def:aN} converges to $f(a)$ defined in Lemma \ref{lemma:astar}. This naturally leads to the following lemma, which establishes the convergence of their minimizers, $\hat{a}_N$ and $a^*$, as well as the corresponding solutions of \eqref{def:fN}, $\hat{\x}=\x_{\hat{a}_N}$ and $\x^*=\x_{a^*}$.
\begin{lemma}\label{converge:a}
 Under Assumption \ref{ass},  the following results hold:
 \begin{enumerate}
 \item[(i)]$\displaystyle\lim_{N\to\infty}\hat{a}_N\overset{a.s.}{=}a^*,$ where $\hat{a}_N$ is the solution to \eqref{def:aN} and $a^*$ is given in Lemma \ref{lemma:astar};
\item[(ii)]
$\lim_{N\to\infty}\frac{1}{{N}}{\|\hat{\x}-\x^*\|^2}\overset{a.s.}{=}0,$ where $\hat{\x}$ is the solution to \eqref{eq:MSErelax} and $\x^{*}$ is defined in Proposition \ref{the1}.%\end{enumerate}
\end{enumerate}
\end{lemma}
The desired result in Proposition \ref{the1} follows immediately from  Lemma \ref{converge:a}.
  \subsection{Proof of Proposition \ref{pro:convergexa}}
  In this part, we provide the idea for proving Proposition \ref{pro:convergexa}.  Recall the notation $(\tau_*^2,\gamma_*)=(\tau_{a^*}^2, \gamma_{a^*})$, and let $\{(\x_{t+1},\bz_t)\}_{t\geq 0}$ be the sequence generated by the AMP algorithm with $\eta_t(x)=\eta_{a^*}(x;\gamma_{*})$.  
We define the following post-processing step:
   \begin{subequations}\label{postprocessing}
 \begin{align}
 \tilde{\x}_{t+1}&=\text{sgn}(\x_{t+1}),\label{xq}\\
\tilde{\bz}_{t+1}&=\bs-\bH\tilde{\x}_{t+1}+\frac{1}{\delta}\left<(\text{sgn}\circ\eta_{a^*})'(\bx_t+\bH^T\bz_t;\gamma_{*})\right>\bz_t.\label{zq}
 \end{align}
 \end{subequations}
 Using the relation that $\x_{t+1}=\eta_{a^*}(\x_t;\gamma_{*})$, \eqref{postprocessing} can be regarded as performing an AMP iteration on $(\x_t,\bz_t)$ with $\eta_t(x)=\text{sgn}\circ\eta_{a^*}(x,\gamma_{*})$.   This allows us to analyze the distribution of $\bH \tilde{\x}_{t+1}$ through AMP theory\footnote{
%Rigorously, to apply the state-evolution theory, the denoising function should be Lipschitz continuous, which does not hold for the sign function. Hereafter, we  work with $\text{sgn}(\cdot)$ for simplicity  to illustrate the main idea of the proof.    Our analysis can be made rigorous by smoothing the sign function  (e.g., by convolving it with a scaled mollifier) to obtain a Lipschitz continuous approximation $q_\epsilon(\cdot)$, and then letting $\epsilon\to0$. %This smoothing argument is standard in the literature and we omit the details for simplicity. 
 Strictly speaking, the state evolution theory requires a Lipschitz denoiser, which the sign function is not. For simplicity, we work with $\text{sgn}(\cdot)$ to illustrate the proof. Rigor can be ensured by smoothing $\text{sgn}(\cdot)$ (e.g., via a mollifier) to obtain a Lipschitz  approximation $q_\epsilon(\cdot)$, and then letting $\epsilon\to0$. 
 }, and the desired prediction for $\bH \text{sgn}(\x^*)$ is then obtained by taking the limit $t \to \infty$.
 
Next, we give a proof sketch of Proposition \ref{pro:convergexa} based  on the above argument.  Following the convention in AMP, let $\tilde{\bb}_{t+1}=\bs-\tilde{\bz}_{t+1}$, $\bb_t=\bs-\bz_t,$ and $\br_{t+1}=\x_t+\bH^T\bz_t$. Denote further $\alpha_t=\frac{1}{\delta}\left<(\text{sgn}\circ\eta_{a^*})'(\br_{t+1};\gamma_{*})\right>$. 
 Then from \eqref{zq}, % we have% \vspace{-0.1cm}
%$$\frac{1}{K}\sum_{k=1}^K\psi\left(\bh_k^T\tilde{\x}_{t+1},s_k\right)%=\frac{1}{K}\sum_{k=1}^K\psi\left(\h_k^T\tilde{\x}_{T+1},s_k\right)
 %=\frac{1}{K}\sum_{k=1}^K\psi\left({\alpha}_ts_k-{\alpha}_t{b}_{t,k}+\tilde{b}_{t+1,k},s_k\right),
%$$
the  goal is to analyze 
   \begin{equation}\label{proof3.2:1}
   \lim_{t\to\infty}\lim_{K\to\infty}\frac{1}{K}\sum_{k=1}^K\psi\left({\alpha}_ts_k-{\alpha}_t{b}_{t,k}+\tilde{b}_{t+1,k},s_k\right).
   \end{equation}
 %$$\bH\tilde{\x}_{t+1}=\alpha_t\bs+\tilde{\bb}_{t+1}-\alpha_t\bb_t.$$
 Before proceeding, we list a few results drawn from  AMP literature that will be used in our derivation; see, e.g., \cite[Lemma 1]{AMP}.
% \begin{enumerate}
 %\item[(i)]  $\mathbf{r}_{t+1}\overset{d}{\approx}\mathcal{N}(\mathbf{0},\tau_t\mathbf{I})$, where $\tau_t$ is defined in \eqref{SE}.
 %\item[(ii)] $\bb_t\overset{d}{\approx}\mathcal{N}(\mathbf{0},\sigma_t\mathbf{I})$, where $\sigma_t^2=\frac{1}{\delta}\mathbb{E}[\eta_a^2(\tau_{t-1}Z;\gamma_a)].$
 %\item[(iii)]$\tilde{\bb}_{t+1}\overset{d}{\approx}\mathcal{N}(\mathbf{0},\tilde{\sigma}_{t+1}\mathbf{I})$,where $\tilde{\sigma}_{t+1}^2=\frac{1}{\delta}\mathbb{E}\left[(\text{sgn}\circ\eta_a)^2(\tau_tZ;\gamma_a)\right]$.
 %\item[(iv)]$\tilde{\bb}_{t+1}^T{\bb}_{t}\approx\tilde{\x}_{t+1}^T{\x}_{t}.$
  %\end{enumerate}
  \begin{lemma}\label{AMP_property}
  The following results hold under Assumption \ref{ass}:
\begin{itemize}
%\item[(i)] For all pseudo-Lipschitz functions $\varphi: \R^2\to\R$ and $t\geq 0$,
%$$\lim_{K\to\infty}\frac{1}{K}\sum_{i=1}^K\varphi(b_t,\tilde{b}_{t+1,i})=\mathbb{E}[\varphi(\sigma_tZ,\tilde{\sigma}_{t+1}Z)],$$
%where 
%\begin{equation}\label{tilde_sigma}
%\tilde{\sigma}_{t+1}^2=\frac{1}{\delta}\mathbb{E}\left[(q_\epsilon\circ\eta_a)^2(\tau_tZ;\gamma_a)\right],
%\end{equation}$Z\sim\mathcal{N}(0,1)$, and $\tau_t^2$ is defined through the state evolution in \eqref{state:AMP}.
\item[(i)] 
$\lim_{N\to\infty}\frac{1}{N}\sum_{n=1}^N\varphi(r_{t+1,n})=\mathbb{E}[\tau_t Z]$ for any pseudo-Lipschitz function $\varphi$, 
where $Z\sim\mathcal{N}(0,1)$, $\tau_t$ is defined in \eqref{SE}  with $a=a^*$ and satisfies $\lim_{t\to\infty}\tau_t^2=\tau_{*}^2$; 
\item[(ii)]
$\lim_{K\to\infty}\langle \bb_t,\tilde{\bb}_{t+1}\rangle=\frac{1}{\delta}\lim_{N\to\infty}\langle \x_t,\tilde{\x}_{t+1}\rangle$, where $\left<\mathbf{u},\mathbf{v}\right>$ takes the inner product of $\mathbf{u}$ and $\mathbf{v}$ normalized by dimension;
\item[(iii)]
$\lim_{K\to\infty}\frac{1}{K}\sum_{i=1}^K\varphi(b_{t,k},\tilde{b}_{t+1,k},s_k)=\mathbb{E}[\varphi(Z_t,\tilde{Z}_{t+1},S)]$  for any pseudo-Lipschitz function $\varphi$,
where $(Z_{t},\tilde{Z}_{t+1})$ are jointly Gaussian distributed and are independent of $S$, with $Z_{t}\sim\mathcal{N}(0,\sigma_t^2)$ and $\tilde{Z}_{t+1}\sim\mathcal{N}(0,\tilde{\sigma}_{t+1}^2)$, $\sigma_t^2=\frac{1}{\delta}\mathbb{E}[\eta_{a^*}^2(\tau_{t-1}Z;\gamma_{*})],$
and $\tilde{\sigma}_{t+1}^2=\frac{1}{\delta}$.%\mathbb{E}\left[(\text{sgn}\circ\eta_{a^*})^2(\tau_tZ;\gamma_{a^*})\right]=\frac{1}{\delta}$. 
\end{itemize}
\end{lemma}
By Lemma \ref{AMP_property} (i) and the Stein's Lemma, we have 
$$
\begin{aligned}\lim_{t\to\infty}\alpha_t&%=\frac{1}{\delta}\mathbb{E}[(\text{sgn}\circ\eta_{a^*})'(\tau_{*}Z;\gamma_{*})]\\
%&
=\frac{1}{\delta\tau_{*}}\mathbb{E}[Z\text{sgn}(\eta_{a^*}(\tau_{*}Z;\gamma_{*}))]=\bar{\alpha}.
\end{aligned}$$
Combining this with \eqref{proof3.2:1} and Lemma \ref{AMP_property} (iii),  it suffices to focus on 
$
%\begin{aligned}
%&\lim_{t\to\infty}\lim_{K\to\infty}\frac{1}{K}\sum_{k=1}^K\psi\left(\bar{\alpha}s_k-\bar{\alpha}{b}_{t,k}+{b}_{t+1,k},s_k\right)\\
\lim_{t\to\infty}\mathbb{E}\left[\psi(\bar{\alpha} S-\bar{\alpha}Z_t+\tilde{Z}_{t+1},S)\right].$
%\end{aligned}$$
Since $(Z_t,\tilde{Z}_{t+1})$ are jointly Gaussian distributed, the distortion term $\bar{Z}_t:=-\bar{\alpha}Z_t+\tilde{Z}_{t+1}$ is also Gaussian; specifically,  $\bar{Z}_t\sim\mathcal{N}(0,\text{var}_t)$ with 
$\text{var}_t=\bar{\alpha}^2\sigma_t^2+\tilde{\sigma}_{t+1}^2-2\bar{\alpha}\mathbb{E}[Z_t\tilde{Z}_{t+1}].$
 The remaining task is to show that $\lim_{t\to\infty}\text{var}_t=\bar{\beta}.$
 From Lemma \ref{AMP_property} (iii) and the definition of $\hat{X}$ in Theorem \ref{theorem1}, we have 
 $\lim_{t\to\infty}\sigma_t^2=\frac{1}{\delta}\mathbb{E}[\hat{X}^2]$ and $\lim_{t\to\infty}\tilde{\sigma}_{t+1}^2=\frac{1}{\delta}.$
In addition, 
 $$
 \begin{aligned}
 \lim_{t\to\infty}\mathbb{E}[Z_t\tilde{Z}_{t+1}]%&\overset{(a)}{=}\lim_{t\to\infty}\lim_{K\to\infty}\langle \bb_t,\tilde{\bb}_{t+1}\rangle\\
 &\overset{(a)}{=}\frac{1}{\delta}\lim_{t\to\infty}\lim_{N\to\infty}\langle \x_t,\tilde{\x}_{t+1}\rangle\\
 &\overset{(b)}{=}\,\frac{1}{\delta}\,\lim_{t\to\infty}\lim_{N\to\infty}\langle \eta_{a^*}(\br_t;\gamma_{*}),\text{sgn}(\br_{t+1})\rangle\\
 &\overset{(c)}{=}\frac{1}{\delta}\mathbb{E}[\eta_{a^*}(\tau_{*} Z;\gamma_{*})\text{sgn}(\tau_{*} Z)]\overset{(d)}{=}\frac{1}{\delta}\mathbb{E}[\hat{X}],
 \end{aligned}$$
 where (a) applies Lemma \ref{AMP_property} (ii) and (iii), (b) holds since  $\text{sgn}\circ\eta_{a^*}(x;\gamma_{*})=\text{sgn}(x)$, (c) follows from Lemma \ref{AMP_property} (i) and the fact that $\lim_{K\to\infty}\lim_{N\to\infty}\frac{1}{N}\|\br_{t+1}-\br_t\|^2\overset{a.s.}{=}0$ (which holds since the AMP algorithm converges), and (d) uses $\eta_{a^*}(x;\gamma_{*}) \text{sgn}(x)=|\eta_{a^*}(x;\gamma_{*})|$ and the definition of $\hat{X}$. 
  Combining the above gives $\lim_{t\to\infty}\text{var}_t=\bar{\beta},$ which completes the proof of Proposition \ref{pro:convergexa}. 
  
  \vspace{-0.1cm}
  \section{Simulation Results}
  \vspace{-0.1cm}
  \begin{figure}
\includegraphics[width=0.3\textwidth]{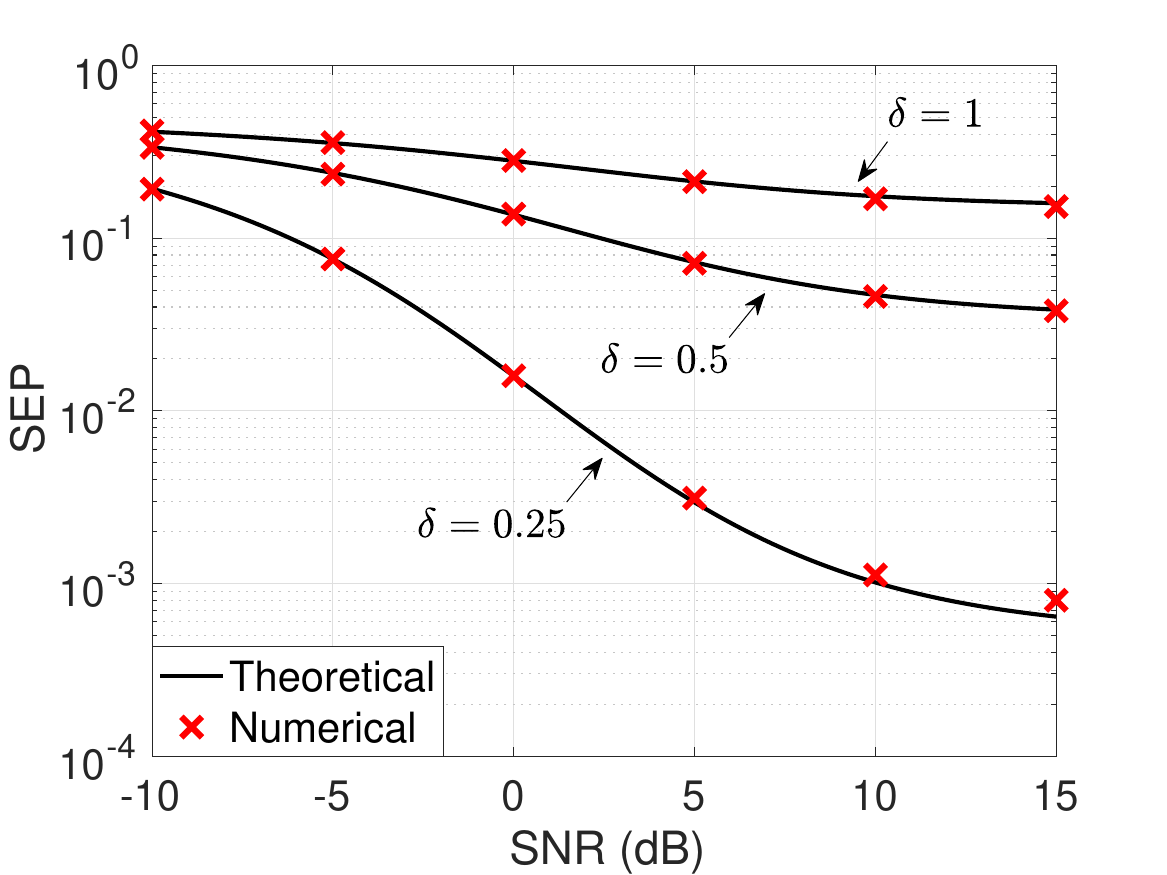}\label{SEPb}
\centering
\vspace{-0.1cm}
\caption{Theoretical and numerical SEP versus the SNR for different systems, where $\rho=\lambda=0.2$.  The number of transmit antennas is fixed as $N=128$, and the number of users is $K=\delta N$.} %The numerical results are averaged over $10^4$ channel realizations.}
\label{fig_SEP}
\end{figure}
% As illustrated, 
In this section, we present simulation results to validate the asymptotic analysis.  In Fig. \ref{fig_SEP}, we plot both the numerical and theoretical SEP (given by Corollary \ref{SEP}) versus the SNR for different systems.   As shown, the numerical and theoretical SEP match closely. %Hence, the theoretical SEP provides a precise estimate of the SEP for practical finite-dimensional systems.  

\begin{figure}
\subfigure[SEP versus $\lambda$, where $\rho=0$.]{\includegraphics[width=0.245\textwidth]{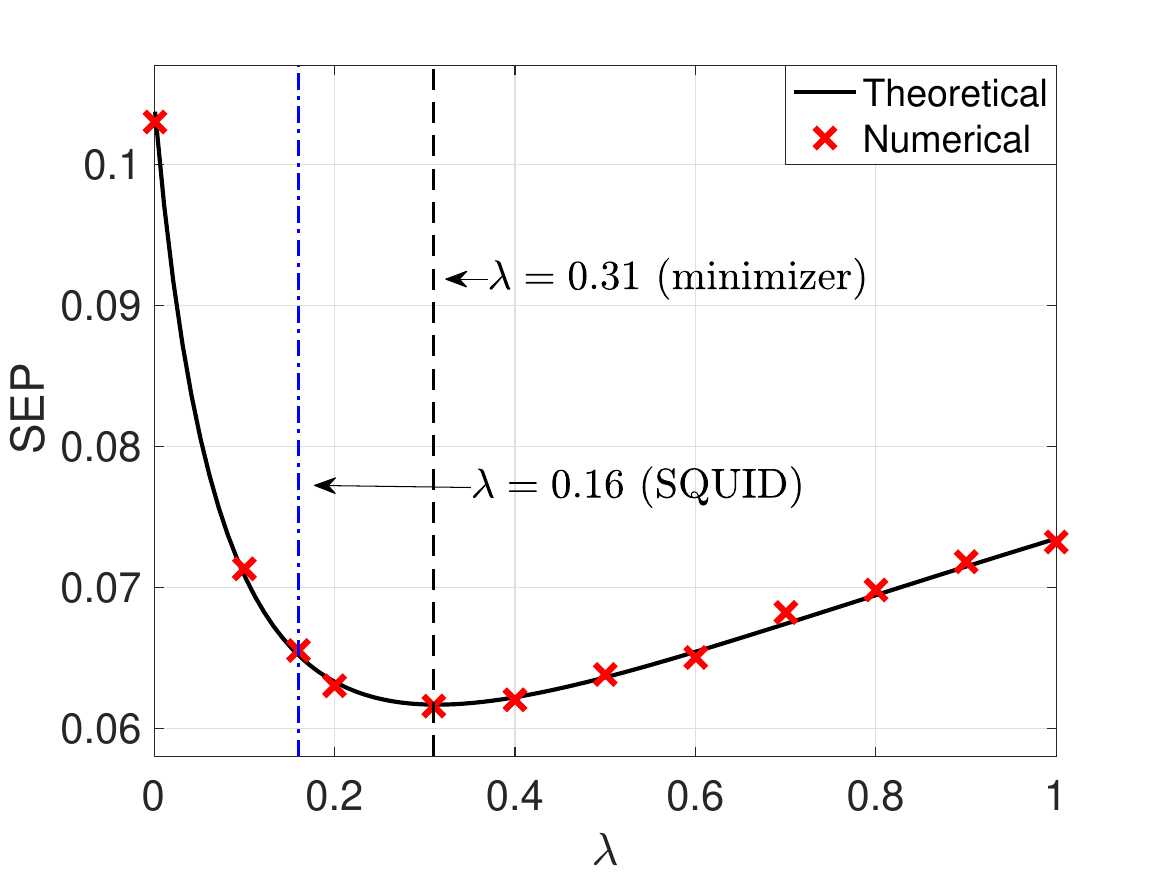}\label{lambda}}\hspace{-0.3cm}
\subfigure[SEP versus $\rho$, where $\lambda=0.31$.]{\includegraphics[width=0.245\textwidth]{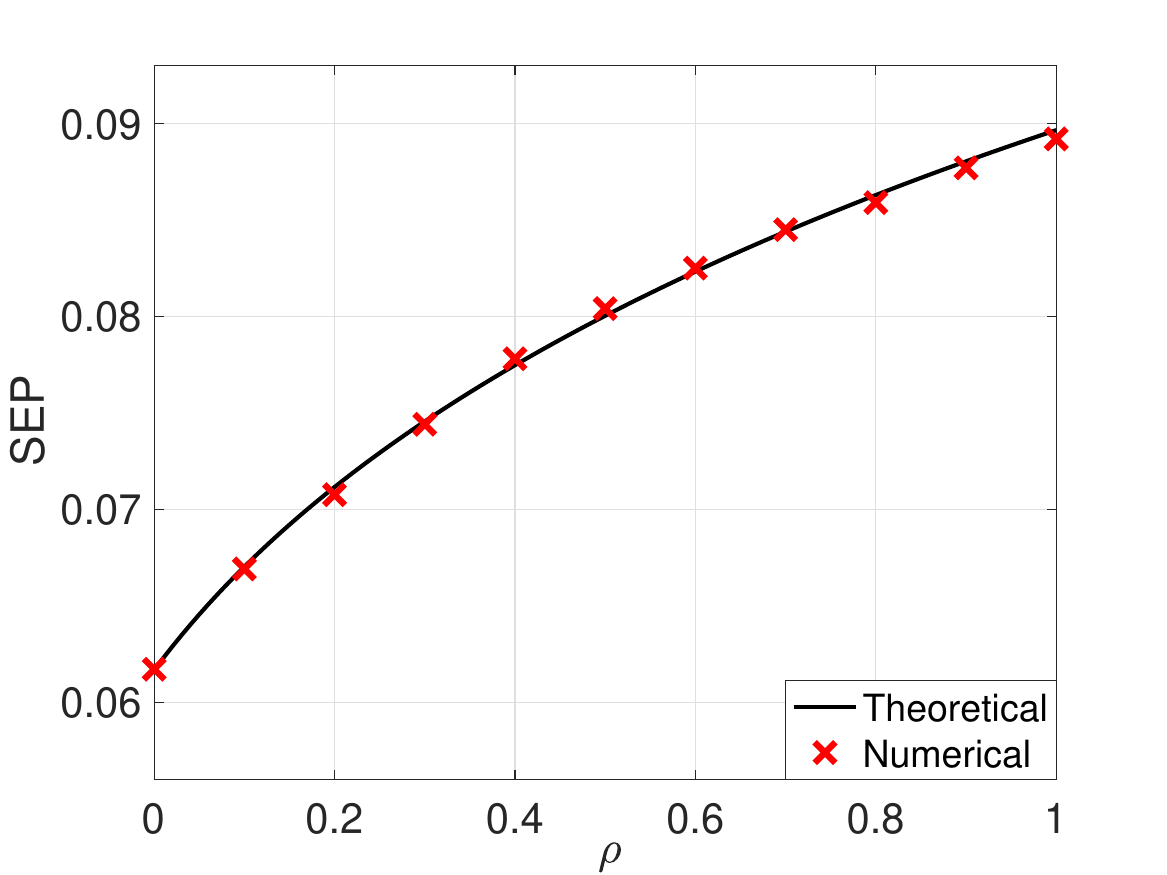}\label{rho}}
\centering
\vspace{-0.1cm}
\caption{Theoretical and numerical SEP versus the regularization parameters $(\rho,\lambda)$, where $\delta=0.5$, $N=128$, and SNR$=5$\,dB.}
\label{fig_parameter}
\end{figure}

Fig. \ref{fig_parameter} further shows the theoretical and numerical SEP versus the regularization parameters $(\rho,\lambda)$. To compare with the classical SQUID precoder \cite{SQUID}, we set $\rho=0$ in Fig. \ref{lambda}. It can be observed from Fig. \ref{lambda} that the asymptotic result provides an accurate prediction when $\rho=0$, although our analysis is rigorous only for positive $\rho$. Moreover, by tuning the regularization parameters, we can achieve additional performance gain over SQUID; for example, setting $\lambda=0.31$ yields better performance than SQUID in the considered system.  In Fig. \ref{rho}, we fix $\lambda=0.31$ and vary $\rho$. As shown, the SEP is minimized at $\rho=0$, i.e., when no $\ell_2$ regularization term is included in the relaxation model \eqref{eq:MSErelax}. % It is worth noting that we have also searched over the two-dimensional $(\rho,\lambda)$ space for various system configurations using our analytical result, and interestingly, the SEP is always minimized at $\rho=0$. A rigorous theoretical justification of this observation is left as future work.
% Based on \eqref{asymptotic}, a natural choice of the scaling factor at the user side is 
% \begin{equation}\label{deterministic:xi}
% \xi=\bar{\alpha}_0^{-1},
% \end{equation}
% which removes the channel gain of the asymptotic model in \eqref{asymptotic}. 

% Note that $\xi$ in \eqref{deterministic:xi} is a deterministic quantity that depends only on the system parameters  $(\delta,\sigma^2)$ and the regularization parameters $(\lambda,\rho)$. This eliminates the need for estimating either the data symbols $\s$ or channel matrix $\bH$ at the user side and thus is more practical than the choice in \eqref{xi}.

Interestingly, by searching over the two-dimensional $(\rho,\lambda)$ space for various system configurations using our analytical result, we observe that the theoretical SEP is always minimized at $\rho=0$ (the results are omitted for space limitations). A rigorous characterization of this observation, or more generally, the joint optimization of the regularization parameters  $(\rho,\lambda)$ based on our analytical results, is left for future work.
\bibliographystyle{IEEEtran}
\bibliography{reference_dce}

\end{document}